\documentclass[11pt]{article}
\usepackage[a4paper,margin=1in]{geometry}
\usepackage[T1]{fontenc}
\usepackage[utf8]{inputenc}
\usepackage{lmodern}
\usepackage{amsmath,amssymb,amsthm,mathtools}
\usepackage{booktabs}
\usepackage{graphicx}
\usepackage[hidelinks]{hyperref}
\usepackage[nameinlink,noabbrev]{cleveref}
\usepackage{enumitem}
\usepackage{microtype}

\newtheorem{theorem}{Theorem}
\newtheorem{proposition}{Proposition}
\newtheorem{lemma}{Lemma}
\theoremstyle{definition}
\newtheorem*{problem}{Problem}
\theoremstyle{plain}
\DeclareMathOperator{\conv}{conv}
\DeclareMathOperator{\Per}{Per}
\crefname{theorem}{Theorem}{Theorems}
\crefname{proposition}{Proposition}{Propositions}
\crefname{lemma}{Lemma}{Lemmas}
\newcommand{\appref}[1]{\hyperref[#1]{Appendix~\ref*{#1}}}

\title{Sibley's Guard-Point Convexity Measure:\\A Perimeter Counterexample and a Dominance Bound}
\author{Masahito Nakano\\Independent researcher, Japan\\\href{mailto:math@nakano-masahito.jp}{\texttt{math@nakano-masahito.jp}}}
\date{August 2026}

\begin{document}
\maketitle

\begin{abstract}
For a simple polygonal region $F$, let $K(F)$ be its polygon kernel, or Sibley's guard-point set, and let $C=\conv(F)$.  The associated guard-point, exterior, and perimeter measures are $G(F)=|K(F)|/|F|$, $E(F)=|F|/|C|$, and $P(F)=\Per(C)/\Per(F)$.  Using a kernel-adapted anisotropic perimeter, we prove $G(F)\le E(F)$.  We disprove the pointwise inequality $G(F)\le P(F)$ by an explicit nonconvex pentagon with integer coordinates for which $G(F)=62/63$ and $P(F)=185/189$.  Nevertheless, $G(F)\le2P(F)$ holds for every simple polygon, and hence $G$ cannot asymptotically dominate $P$ in Sibley's sense.  Thus the two assertions in Sibley's Conjecture~2 are settled in opposite directions.  We also observe that Sibley's convexity coefficient $\chi$ and interior measure $I$ coincide, respectively, with the Beer index $b$ and convexity ratio $c$.  The theorem $b\le180c$ of Balko, Jel\'inek, Valtr, and Walczak therefore yields $\chi\le180I$, resolving Sibley's Conjecture~1.
\end{abstract}



\section{Introduction}
Convexity is a binary property, but non-convex sets may be closer to convex in several different geometric senses.  Sibley \cite{Sibley2025} compared six normalized measures for polygons together with their interiors and formulated two numbered conjectures concerning comparisons among them.  Conjecture~2 was subsequently highlighted as Question~2, ``Thomas Sibley's Conjectures,'' in the CCCG 2025 open-problems summary \cite{CCCG2025}.  The present paper uses Sibley's notation and terminology.

Let $F\subset \mathbb{R}^2$ be a simple polygonal region, including its boundary and interior, and let
\[
C=\conv(F)
\]
be its convex hull.  A point $g\in F$ is a \emph{guard point} if
\[
[g,p]\subset F\qquad\text{for every }p\in F,
\]
where $[g,p]$ denotes the closed segment joining $g$ and $p$.  We write
\[
K(F)=\{g\in F:[g,p]\subset F\text{ for every }p\in F\}.
\]
In standard visibility terminology, $K(F)$ is the kernel of $F$.  Following Sibley's terminology, we call $K(F)$ the set of guard points, or the guard-point set, of $F$.  Sibley proves that this set is convex \cite{Sibley2025}.  For a counterclockwise simple polygon, $K(F)$ is equivalently the intersection of the closed interior half-planes determined by the oriented sides of $F$.

The three measures considered here are
\[
G(F)=\frac{|K(F)|}{|F|},\qquad
E(F)=\frac{|F|}{|C|},\qquad
P(F)=\frac{\Per(C)}{\Per(F)},
\]
where $|\cdot|$ denotes area and $\Per(\cdot)$ denotes ordinary Euclidean perimeter.  These are Sibley's notations; in the CCCG open-problems summary, the corresponding exterior and perimeter measures are denoted by $A$ and $L$, respectively.  For a simple polygonal curve, each hull side replaces a boundary chain by the corresponding chord, so $\Per(C)\le \Per(F)$ and hence $0<P(F)\le 1$.  If $|K(F)|=0$, then $G(F)=0$.  Since $E(F)>0$ and $P(F)>0$ for every simple polygonal region, both upper bounds for $G$ proved below, namely $G(F)\le E(F)$ and $G(F)\le2P(F)$, then hold automatically.  It therefore remains to consider the substantive case $|K(F)|>0$; in particular, such a polygon is star-shaped.

Sibley compares convexity measures not only by pointwise inequalities, but also by an asymptotic relation called domination.  In Sibley's sense, a measure $B$ \emph{dominates} a measure $A$, written $B\gg A$, if there exists a sequence of polygons $F_n$ such that
\[
B(F_n)\to 1,\qquad A(F_n)\to 0.
\]
Thus $B\gg A$ means that a sequence may look nearly convex according to $B$ while becoming extremely non-convex according to $A$.  We write $B\not\gg A$ for the negation of this relation; in particular, $G\not\gg P$ means that no sequence of polygons can satisfy $G(F_n)\to 1$ while $P(F_n)\to 0$.  The two pointwise inequalities in Sibley's Conjecture 2,
\[
G(F)\le E(F),\qquad G(F)\le P(F),
\]
would immediately imply $G\not\gg E$ and $G\not\gg P$.  They are therefore strong pointwise sufficient conditions for two entries in Sibley's domination comparison.

Our first purpose is to resolve Conjecture~2 while separating its pointwise assertions from the underlying domination intuition.  We prove that
\[
G\le E\quad\text{is true},
\]
show by an explicit pentagon that
\[
G\le P\quad\text{is false},
\]
and nevertheless prove the uniform estimate
\[
G(F)\le 2P(F),
\]
which implies
\[
G\not\gg P.
\]
Thus the two pointwise assertions in Conjecture~2 are settled in opposite directions, while $G$ still cannot asymptotically dominate $P$ in Sibley's sense.  These are the new geometric results of this paper.

Our second purpose is bibliographic.  Sibley's Conjecture~1 asks for an absolute constant $r$ such that
\[
\chi(F)\le rI(F)
\]
for every simple polygon.  We identify $\chi$ and $I$, respectively, with the previously studied Beer index $b$ and convexity ratio $c$.  Corollary~1.4 of Balko, Jel\'inek, Valtr, and Walczak \cite{BalkoEtAl2017} then implies Conjecture~1 with $r=180$.  This identification and its consequence are recorded in \cref{sec:beer}; the estimate itself is due to Balko et al.

\begin{figure}[t]
\centering
\includegraphics[width=0.93\linewidth,trim=0 8.5pt 0 12pt,clip]{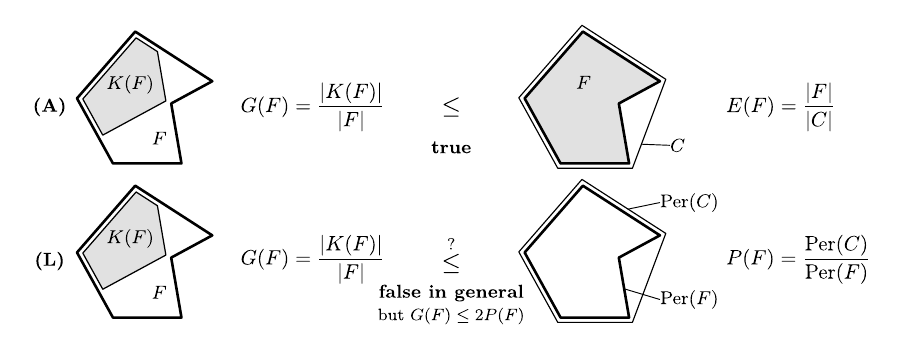}
\caption{The two pointwise comparisons in Sibley's Conjecture~2.  The diagram records the universal conclusions, not the numerical values of the depicted shape: $G\le E$ holds; $G\le P$ fails, but $G\le2P$ holds.  Coincident hull and polygon boundaries are slightly offset for visibility.}
\label{fig:comparisons}
\end{figure}

\section{Definitions and main results}
The following theorem collects the main results concerning Sibley's original measures.

\begin{theorem}\label{thm:main}
The following statements hold for Sibley's guard-point measure $G$, exterior measure $E$, and perimeter measure $P$.
\begin{enumerate}[label=\textup{(\roman*)},leftmargin=*]
\item For every simple polygon $F$,
\[
G(F)\le E(F).
\]
\item The pointwise perimeter inequality $G(F)\le P(F)$ fails in general.  More precisely, there exists a simple nonconvex pentagon $F$ for which
\[
G(F)=\frac{62}{63},\qquad P(F)=\frac{185}{189},
\]
so that $G(F)>P(F)$.
\item For every simple polygon $F$,
\[
G(F)\le 2P(F).
\]
\end{enumerate}
\end{theorem}

The bounds in \Cref{thm:main}\textup{(i),(iii)} imply $G\not\gg E$ and $G\not\gg P$, respectively.  In other words, no sequence of simple polygons can have $G(F_n)\to1$ while either $E(F_n)\to0$ or $P(F_n)\to0$.

\section{An anisotropic perimeter for the exterior inequality}\label{sec:exterior}

The proof of the exterior inequality uses an anisotropic perimeter determined by the guard-point set.  Let
\[
R=\begin{pmatrix}0&1\\-1&0\end{pmatrix},\qquad R(x,y)=(y,-x),
\]
so that $R$ is clockwise rotation by $90^\circ$.  If a polygonal boundary is oriented counterclockwise and an edge has side vector $s_i=p_{i+1}-p_i$, then $n_i=Rs_i$ is the outward normal vector to that edge, with $\|n_i\|=\|s_i\|$.  For a compact convex set $A$, let
\[
h_A(u)=\max_{x\in A}x\cdot u
\]
be its support function.  If $Q$ is a polygonal region whose boundary is oriented counterclockwise, define the $A$-anisotropic perimeter of $Q$ by
\[
\Per_A(Q)=\sum_i h_A(n_i),
\]
where $n_i$ is the outward normal vector obtained by rotating the $i$-th side vector clockwise by $90^\circ$.  For our purposes this finite sum is the definition of $\Per_A(Q)$.  Equivalently, if $\nu_Q$ denotes the outward unit normal on each side, then each side contributes $h_A(\nu_Q)$ times its Euclidean length.  Thus the displayed sum is the polygonal version of the usual anisotropic perimeter; no measure-theoretic formulation is needed below.

In most applications below, the anisotropy body $A$ will be the guard-point set $K=K(F)$, and we then write $\Per_K(\cdot)$ for $\Per_{K(F)}(\cdot)$.  We will also apply the same framework to an inscribed disk in $K(F)$.  The functional $\Per_A(\cdot)$ is the anisotropic perimeter associated with the auxiliary body $A$; the ratio introduced below is a guard-point-adapted perimeter ratio rather than a new member of Sibley's original list.

Although the support function depends on the choice of origin, $\Per_A(Q)$ is unchanged when the anisotropy body $A$ is translated; the verification is given in \appref{app:anisotropic}.  Thus, when $|K(F)|>0$, we may normalize the support function by translating only its anisotropy body so that $0\in\operatorname{int}K$.  This leaves every quantity $\Per_K(Q)$ unchanged and ensures $h_K(u)>0$ for every nonzero $u$.

The same functional also has the affine covariance that is absent from ordinary Euclidean perimeter.  The proof is given in \appref{app:anisotropic}.

\begin{proposition}[Affine covariance]\label{prop:affine}
Let $\Phi$ be a nonsingular affine transformation with linear part $T\in\mathrm{GL}(2,\mathbb{R})$.  For every compact convex polygon $A$ and every polygonal region $Q$,
\[
\Per_{\Phi(A)}(\Phi(Q))=|\det T|\,\Per_A(Q).
\]
Moreover,
\[
K(\Phi(F))=\Phi(K(F)),\qquad
\conv(\Phi(F))=\Phi(\conv(F)).
\]
Consequently, $G$, $E$, and the kernel-adapted ratio $P_K$ defined below are invariant under nonsingular affine transformations.
\end{proposition}

When $K=K(F)$ has positive area, define the guard-point-adapted perimeter ratio
\[
P_K(F)=\frac{\Per_K(K)}{\Per_K(F)}.
\]
Equivalently, this ratio could be written as $P_{K(F)}(F)$; after fixing $F$ we write $K=K(F)$ and use the shorter notation $P_K(F)$.  After the harmless normalization $0\in\operatorname{int}K$, the support function is positive on every nonzero edge-normal vector, and hence $\Per_K(F)>0$; thus the denominator in this ratio is positive.  The same framework that proves $G\le E$ also gives the following restored pointwise control for this anisotropic analogue.

\begin{proposition}\label{prop:PK}
Let $F$ be a simple polygon whose guard-point set $K=K(F)$ has positive area.  Then
\[
P_K(F)^2\le G(F)\le P_K(F).
\]
\end{proposition}

The point of \Cref{prop:PK} is explanatory.  It isolates the directional information lost by ordinary Euclidean perimeter and shows that the expected pointwise control is restored for the guard-point-adapted anisotropic ratio.

We first record three geometric inputs.  Their proofs, including the sign convention connecting the clockwise rotation with the shoelace formula, are collected in \appref{app:anisotropic}.  The applications below only require $|K(F)|>0$; when $|K(F)|=0$, the inequalities with left-hand side $G(F)$ are immediate.

\begin{lemma}\label{lem:perKK}
If $K$ is a convex polygon with positive area, then
\[
\Per_K(K)=2|K|.
\]
\end{lemma}

\begin{lemma}\label{lem:perAF}
Let $F$ be a simple polygon, and let $A\subseteq K(F)$ be a nonempty compact convex set.  Then
\[
\Per_A(F)\le 2|F|.
\]
\end{lemma}

\begin{lemma}\label{lem:convhull}
Let $F$ be a simple polygon with nonempty guard-point set $K=K(F)$, and let $C=\conv(F)$.  Then
\[
\Per_K(C)\le \Per_K(F).
\]
\end{lemma}

Geometrically, each hull side replaces a boundary chain of $F$, and the subadditivity of the support function makes this replacement non-increasing for the $K$-anisotropic perimeter.

We also use the following polygonal form of the planar Wulff inequality.  The short derivation from the planar Brunn--Minkowski theorem and the mixed-area formula is deferred to \appref{app:wulff}.

\medskip
\noindent\textbf{Planar Wulff inequality.}  For convex polygons $K,Q\subset\mathbb{R}^2$, with $K$ of positive area,
\[
\Per_K(Q)^2\ge 4|K|\,|Q|.
\]
Equality holds when $Q$ is homothetic to $K$.

\begin{proof}[Proof of $G(F)\le E(F)$]
If $|K|=0$, then $G(F)=0$.  Assume $|K|>0$ and use the translation normalization above, so the anisotropic perimeters below are nonnegative.  Applying the Wulff inequality to $Q=C=\conv(F)$ gives
\[
\Per_K(C)^2\ge 4|K|\,|C|.
\]
By \cref{lem:convhull,lem:perAF}, applied with $A=K$,
\[
\Per_K(C)\le \Per_K(F)\le 2|F|.
\]
All anisotropic perimeters in this chain are nonnegative, so we may combine the inequalities in squared form:
\[
4|K|\,|C|\le \Per_K(C)^2\le \Per_K(F)^2\le 4|F|^2.
\]
Dividing by $4|F|\,|C|$ yields
\[
\frac{|K|}{|F|}\le \frac{|F|}{|C|},
\]
which is exactly $G(F)\le E(F)$.
\end{proof}

Convex polygons satisfy equality in this inequality, since then $K=F=C$.  We do not attempt to characterize all equality cases.  Indeed, the equality conditions visible from this proof would have to hold simultaneously in the Wulff inequality, in the convex-hull replacement $\Per_K(C)\le \Per_K(F)$, and in the estimate $\Per_K(F)\le 2|F|$.  This makes the complete equality problem substantially more delicate than the pointwise inequality proved here.

\begin{proof}[Proof of \Cref{prop:PK}]
Use the same translation normalization, so that $\Per_K(F)>0$ while all anisotropic perimeters retain their original numerical values.

By \cref{lem:perKK,lem:perAF}, applied with $A=K$,
\[
P_K(F)=\frac{\Per_K(K)}{\Per_K(F)}=\frac{2|K|}{\Per_K(F)}\ge \frac{2|K|}{2|F|}=G(F).
\]
This proves $G(F)\le P_K(F)$.  For the other inequality, the Wulff inequality and \cref{lem:convhull} give
\[
\Per_K(F)\ge \Per_K(C)\ge 2\sqrt{|K|\,|C|}.
\]
Therefore
\[
P_K(F)=\frac{2|K|}{\Per_K(F)}\le \sqrt{\frac{|K|}{|C|}}=\sqrt{G(F)E(F)}.
\]
Since $E(F)\le 1$, this gives $P_K(F)\le \sqrt{G(F)}$.  Hence $P_K(F)^2\le G(F)$.
\end{proof}

\section{A pentagonal counterexample to \texorpdfstring{$G(F)\le P(F)$}{G(F) <= P(F)}}\label{sec:counterexample}
We now give the promised counterexample.  The coordinates are kept at an integer scale in order to make all relevant lengths integral; the relevant lengths come from two Pythagorean triples after their common factors are exposed.

Let $F$ be the simple polygon whose boundary vertices, in counterclockwise cyclic order, are
\[
v_1=(0,4620),\quad v_2=(0,-4620),\quad v_3=(23100,-385),\quad v_4=(22176,0),\quad v_5=(23100,385).
\]
A direct check shows that $v_4$ is the unique reflex vertex.  The convex hull is
\[
C=\conv(F)=\operatorname{poly}(v_1,v_2,v_3,v_5).
\]
Here $\operatorname{poly}(p_1,\dots,p_m)$ denotes the polygonal region with the listed vertices in cyclic order.  The kernel is
\[
\begin{aligned}
K(F)=\operatorname{poly}\bigl(& (0,-4620),(19800,-990),(22176,0),\\
                              & (19800,990),(0,4620)\bigr),
\end{aligned}
\]
and the exact values are summarized in \Cref{tab:counterexample-data}.  The half-plane calculation, shoelace computations, and integral side-length verification are collected in \appref{app:counterexample}.

\begin{table}[ht]
\centering
\caption{Exact data for the pentagonal counterexample.}
\label{tab:counterexample-data}
\begin{tabular}{@{}lrlr@{}}
\toprule
Quantity & Exact value & Quantity & Exact value\\
\midrule
$|F|$ & $115259760$ & $\Per(F)$ & $58212$\\
$|K(F)|$ & $113430240$ & $\Per(C)$ & $56980$\\
$|C|$ & $115615500$ & $G(F)$ & $62/63$\\
$P(F)$ & $185/189$ & $E(F)$ & $324/325$\\
\bottomrule
\end{tabular}
\end{table}

In particular,
\[
G(F)-P(F)=\frac1{189}>0.
\]
This proves that Sibley's pointwise perimeter inequality $G(F)\le P(F)$ is false.

\begin{figure}[ht]
\centering
\includegraphics[width=\linewidth,trim=0 8.5pt 0 6pt,clip]{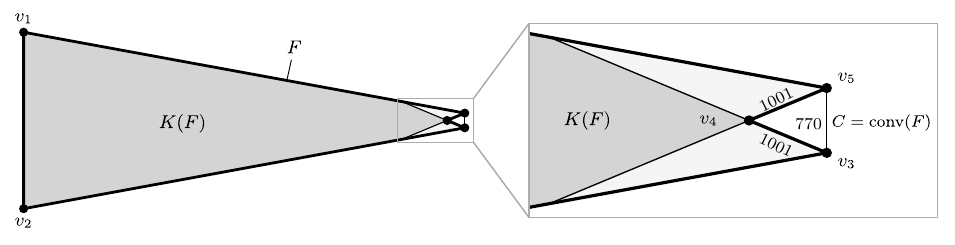}
\caption{The pentagonal counterexample.  Left: the full polygon $F$ and its kernel $K(F)$.  Right: an enlarged view of the boxed notch; $v_4$ is the unique reflex vertex, and the thin segment $v_3v_5$ is the hull edge.  The heavy curve is the polygonal boundary of $F$, and filling the notch gives $C=\conv(F)$.  The displayed lengths are in the original coordinate scale.}
\label{fig:counterexample}
\end{figure}

\section{The ordinary perimeter comparison}\label{sec:ordinary}
The counterexample shows that the pointwise inequality $G\le P$ is false.  It does not, however, imply that $G$ dominates $P$.  In fact the ordinary perimeter measure still controls $G$ up to the universal factor $2$.  We state the three elementary Euclidean estimates needed for the proof; their proofs are given in \appref{app:euclidean}.  The first estimate is a direct specialization of the anisotropic bound in \cref{lem:perAF} to a disk contained in the kernel.

\begin{lemma}\label{lem:fan}
Let $F$ be a simple polygon and suppose that $K=K(F)$ contains a disk $B(o,\rho)$ with $\rho>0$.  Then
\[
\Per(F)\le \frac{2|F|}{\rho}.
\]
\end{lemma}

\begin{lemma}\label{lem:permonotone}
If $A\subset B$ are compact convex polygonal regions with nonempty interior, then
\[
\Per(A)\le \Per(B).
\]
\end{lemma}

\begin{lemma}\label{lem:inradius}
If $Q$ is a compact convex polygonal region with nonempty interior and inradius $\rho>0$, then
\[
|Q|\le \rho\Per(Q).
\]
\end{lemma}

\begin{proof}[Proof of $G(F)\le 2P(F)$]
If $|K(F)|=0$, then $G(F)=0$.  Assume $|K|>0$.  Let $\rho$ be the inradius of $K$, and choose a maximal inscribed disk $B(o,\rho)\subset K$.  By \cref{lem:fan},
\[
\Per(F)\le \frac{2|F|}{\rho}.
\]
On the other hand, the intentionally nonsharp estimate \cref{lem:inradius}, applied to the convex polygon $K$, gives
\[
|K|\le \rho\Per(K).
\]
Since $K\subset C=\conv(F)$, \cref{lem:permonotone} gives
\[
\Per(K)\le \Per(C).
\]
Now compute
\[
\frac{G(F)}{P(F)}=\frac{|K|}{|F|}\cdot \frac{\Per(F)}{\Per(C)}.
\]
Using the upper bound for $\Per(F)$, we obtain
\[
\frac{G(F)}{P(F)}\le \frac{|K|}{|F|}\cdot \frac{2|F|}{\rho\Per(C)}=\frac{2|K|}{\rho\Per(C)}.
\]
Using $|K|\le \rho\Per(K)$ and $\Per(K)\le\Per(C)$, we conclude that
\[
\frac{G(F)}{P(F)}\le2,
\]
or equivalently $G(F)\le2P(F)$.
If $G\gg P$, there would be polygons $F_n$ with $G(F_n)\to 1$ and $P(F_n)\to 0$.  The displayed inequality would force $G(F_n)\le 2P(F_n)\to 0$, a contradiction.  Thus $G\not\gg P$.
\end{proof}

\section{Sibley's Conjecture 1 and the Beer index}\label{sec:beer}

Unlike the preceding results on Conjecture~2, the resolution of Conjecture~1 requires no new geometric estimate.  For $s\in F$, let
\[
V_F(s)=\{t\in F:[s,t]\subset F\}.
\]
Sibley's convexity coefficient is
\[
\chi(F)=\frac{1}{|F|^2}\int_F |V_F(s)|\,ds.
\]
Equivalently, it is the probability that two independent uniformly distributed points of $F$ see each other.  This is precisely the Beer index $b(F)$.  Sibley's interior measure is
\[
I(F)=\frac{1}{|F|}\max\{|Q|:Q\subset F\text{ is convex}\},
\]
which is the convexity ratio $c(F)$.  Thus, for simple polygonal regions,
\[
\chi(F)=b(F),\qquad I(F)=c(F).
\]

Corollary~1.4 of Balko, Jel\'inek, Valtr, and Walczak gives $b(F)\le180c(F)$ for every simple polygon $F$ \cite{BalkoEtAl2017}.  Consequently,
\[
\chi(F)\le180I(F).
\]
This proves Sibley's Conjecture~1 with $r=180$ and implies $\chi\not\gg I$.  The constant $180$ is inherited from the earlier Beer-index estimate; no claim of optimality is made here.

\section{Discussion}
The two pointwise assertions in Sibley's Conjecture~2 behave differently: the exterior inequality $G\le E$ is true, whereas the Euclidean perimeter inequality $G\le P$ is false.  The latter failure does not reverse the corresponding domination conclusion.  The uniform estimate $G(F)\le2P(F)$ still gives $G\not\gg P$.

The contrast between $P$ and $P_K$ is structural.  By \Cref{prop:affine}, $G$, $E$, and $P_K$ are invariant under nonsingular affine transformations, whereas the ordinary Euclidean perimeter ratio $P$ need not be; the family below gives a concrete illustration.  The pointwise comparison
\[
P_K(F)^2\le G(F)\le P_K(F)
\]
shows that the expected control is restored once the directional information encoded by the kernel is retained.

\begin{problem}[Optimal constant]
Determine
\[
\alpha^*=\sup_F\frac{G(F)}{P(F)},
\]
where the supremum is over all simple polygons.
\end{problem}

The upper bound in \Cref{thm:main}\textup{(iii)} gives $\alpha^*\le2$.  As a benchmark, apply the vertical compressions
\[
T_k(x,y)=(x,y/k),
\qquad F_k=T_k(F),
\]
to the pentagon $F$ from \cref{sec:counterexample}.  For every finite $k\ge1$, the map $T_k$ is nonsingular, so $F_k$ is again a simple polygon.  By \Cref{prop:affine},
\[
G(F_k)=G(F)=\frac{62}{63}.
\]
The side lengths give
\begin{align*}
\Per(F_k)
&=\frac{9240}{k}
  +2\sqrt{23100^2+\left(\frac{4235}{k}\right)^2}
  +2\sqrt{924^2+\left(\frac{385}{k}\right)^2}
  \longrightarrow 48048,\\
\Per(\conv(F_k))
&=\frac{10010}{k}
  +2\sqrt{23100^2+\left(\frac{4235}{k}\right)^2}
  \longrightarrow 46200.
\end{align*}
Hence $P(F_k)\to25/26$.  In particular, $P(F_1)=185/189$ while $P(F_k)\to25/26$, so $P$ is not invariant under nonsingular affine transformations.  Moreover,
\[
\frac{G(F_k)}{P(F_k)}
\longrightarrow
\frac{62}{63}\cdot\frac{26}{25}
=\frac{1612}{1575}
=1.023492063\ldots.
\]
Since $\alpha^*\ge G(F_k)/P(F_k)$ for every $k$, letting $k\to\infty$ gives $\alpha^*\ge1612/1575$.  Consequently,
\[
\frac{1612}{1575}\le\alpha^*\le2.
\]
This refinement is ancillary to the domination result.  The constant $2$ is used here to prove $G\not\gg P$ and is not claimed to be sharp; determining the exact value of $\alpha^*$ remains open.

\appendix
\section*{Appendices}
\section{Proofs for the anisotropic-perimeter estimates}\label{app:anisotropic}
We use the clockwise rotation
\[
R(x,y)=(y,-x).
\]
For a counterclockwise polygonal boundary, rotating an oriented side vector by $R$ gives the outward normal vector with length equal to that side.  The sign convention is compatible with the shoelace formula because
\[
a\cdot Rb=\det(a,b).
\]

The anisotropic perimeter used in the main text is translation invariant in its first argument.  If $A$ is compact and convex and $\tau\in\mathbb{R}^2$, then
\[
h_{A+\tau}(u)=h_A(u)+\tau\cdot u.
\]
For a closed polygonal boundary, the side vectors satisfy $\sum_i e_i=0$, and therefore the corresponding normal vectors satisfy $\sum_i Re_i=0$.  Hence
\[
\Per_{A+\tau}(Q)=\sum_i h_{A+\tau}(Re_i)=\sum_i h_A(Re_i)+\tau\cdot\sum_i Re_i=\Per_A(Q).
\]
This justifies translating the guard-point set when applying the Wulff inequality.

\begin{proof}[Proof of \Cref{prop:affine}]
Write $\Phi(x)=Tx+\tau$.  The translation vector $\tau$ plays no role in the covariance calculation: translating $Q$ does not change its side vectors, while $\Per_{A+\tau}(Q)=\Per_A(Q)$ by the translation invariance established above.
First suppose that $\det T>0$.  If $e_i$ are the counterclockwise oriented side vectors of $Q$, then $Te_i$ are the corresponding side vectors of $\Phi(Q)$.  For the clockwise quarter-turn $R$ used here,
\[
RT=(\det T)T^{-\mathsf T}R.
\]
Translations of the anisotropy body do not change anisotropic perimeter, and the support function satisfies $h_{TA}(u)=h_A(T^{\mathsf T}u)$.  Therefore
\begin{align*}
\Per_{\Phi(A)}(\Phi(Q))
&=\sum_i h_{TA}(RTe_i)\\
&=\sum_i h_A\bigl(T^{\mathsf T}RTe_i\bigr)\\
&=(\det T)\sum_i h_A(Re_i)
=(\det T)\Per_A(Q).
\end{align*}
If $\det T<0$, the transformed cyclic order is clockwise.  Reversing it replaces each transformed side vector by its negative, and the same computation gives the factor $-\det T=|\det T|$.  This proves the covariance formula.

Because an affine map carries every segment onto the segment joining the images of its endpoints,
\[
\Phi([g,p])=[\Phi(g),\Phi(p)].
\]
Since $\Phi$ is bijective, the defining visibility condition for the kernel is therefore preserved in both directions, which gives $K(\Phi(F))=\Phi(K(F))$.  Affine maps also commute with convex hulls, so $\conv(\Phi(F))=\Phi(\conv(F))$.  Moreover, every measurable set $S\subset\mathbb{R}^2$ satisfies
\[
|\Phi(S)|=|\det T|\,|S|.
\]
Thus the numerator and denominator of each of the area ratios $G$ and $E$ acquire the same factor $|\det T|$.  Together with the covariance formula for anisotropic perimeter, this proves that $G$, $E$, and $P_K$ are invariant under nonsingular affine transformations.
\end{proof}

\begin{proof}[Proof of \Cref{lem:perKK}]
Orient the boundary of $K$ counterclockwise, and write its vertices in cyclic order as $q_1,\dots,q_m$.  Put
\[
e_i=q_{i+1}-q_i,
\qquad n_i=Re_i,
\]
with indices taken modulo $m$.  Then $n_i$ is the outward normal vector to the side $[q_i,q_{i+1}]$.  Since this side is a supporting side of $K$,
\[
h_K(n_i)=q_i\cdot n_i.
\]
Moreover,
\[
q_i\cdot n_i=q_i\cdot R(q_{i+1}-q_i)=q_i\cdot Rq_{i+1}=\det(q_i,q_{i+1}),
\]
because $q_i\cdot Rq_i=0$.  Hence
\[
\Per_K(K)=\sum_i h_K(n_i)=\sum_i\det(q_i,q_{i+1})=2|K|.
\]
\end{proof}

\begin{proof}[Proof of \Cref{lem:perAF}]
Orient the boundary of $F$ counterclockwise, and let its vertices be $v_1,\dots,v_m$ in cyclic order.  Put
\[
e_i=v_{i+1}-v_i,
\qquad n_i=Re_i.
\]
The guard-point set is the intersection of the closed inner half-planes determined by the oriented sides of $F$.  With the clockwise-rotation convention used here, $n_i=Re_i$ is the outward normal, so the inner half-plane is exactly the set of points $x$ satisfying $x\cdot n_i\le v_i\cdot n_i$.  Since $A\subseteq K(F)$, every $x\in A$ satisfies
\[
x\cdot n_i\le v_i\cdot n_i.
\]
Taking the maximum over $x\in A$ gives
\[
h_A(n_i)\le v_i\cdot n_i.
\]
Summing over all sides and using the shoelace formula for the simple polygon $F$, we obtain
\[
\Per_A(F)=\sum_i h_A(n_i)\le \sum_i v_i\cdot n_i=\sum_i\det(v_i,v_{i+1})=2|F|.
\]
\end{proof}

\begin{proof}[Proof of \Cref{lem:convhull}]
Choose a cyclic subdivision of the convex-hull boundary whose vertices include all vertices of $F$ lying on $\partial C$.  This subdivision does not change $C$ or its anisotropic perimeter, because the support-function contribution is additive along consecutive collinear sides.  If several consecutive vertices of $F$ lie on the same side of $C$, we include all of them in this subdivision; zero-turn or collinear subdivisions do not affect the sum.  With this subdivision understood, each oriented side of $C$ has endpoints that are vertices of $F$.  Since $F\subset C$ and $\partial F$ is a Jordan curve, the visits of $\partial F$ to the vertices lying on $\partial C$ respect the cyclic order induced by $\partial C$; otherwise two connecting boundary arcs of $F$ inside the disk $C$ would have to cross, contradicting simplicity.

Let $ab$ be one such oriented side of the subdivided boundary of $C$, with side vector $e=b-a$.  We choose the boundary chain lying on the counterclockwise boundary of $F$ between the two consecutive hull vertices $a$ and $b$ in the induced cyclic order; by construction, this chain contains no other vertex of the subdivided hull boundary in its interior.  Replacing this chain by the segment $ab$ gives precisely the corresponding side of the convex hull subdivision.  Write the side vectors of this chain as $e_1,\dots,e_r$.  As $ab$ runs over the oriented sides of the subdivided boundary of $C$, these boundary chains are disjoint except for endpoints and together partition the boundary of $F$.  Hence
\[
e=e_1+\cdots+e_r.
\]
Since $R$ is linear,
\[
Re=Re_1+\cdots+Re_r.
\]
By subadditivity of the support function,
\[
h_K(Re)\le \sum_{j=1}^r h_K(Re_j).
\]
Thus the contribution of the hull side $ab$ to $\Per_K(C)$ is no larger than the sum of the contributions of the corresponding boundary chain of $F$.  Summing over all oriented sides of $C$ gives
\[
\Per_K(C)\le \Per_K(F).
\]
\end{proof}

\section{Exact verification of the counterexample}\label{app:counterexample}
This appendix records the exact arithmetic for the pentagonal counterexample in \cref{sec:counterexample}.

The vertices
\[
v_1=(0,4620),\quad v_2=(0,-4620),\quad v_3=(23100,-385),\quad v_4=(22176,0),\quad v_5=(23100,385)
\]
are in counterclockwise cyclic order.  Therefore the guard-point set is the intersection of the closed left half-planes determined by the oriented sides.  For reference, the following table labels the side constraints used in the calculation.  Except for the harmless simplification of $9240x\ge0$ to $x\ge0$ in $H_{12}$, the coefficients are left unreduced because they are obtained directly from the oriented edge vectors and the corresponding normals defining the left half-planes; reducing them would give the same half-planes but would obscure the integer scale at which the Pythagorean triples enter the construction.

\begin{center}
\begin{tabular}{@{}ll@{}}
\toprule
label and oriented side & closed left half-plane\\
\midrule
$H_{12}: v_1v_2$ & $x\ge 0$\\
$H_{23}: v_2v_3$ & $-4235x+23100y+106722000\ge 0$\\
$H_{34}: v_3v_4$ & $-385x-924y+8537760\ge 0$\\
$H_{45}: v_4v_5$ & $-385x+924y+8537760\ge 0$\\
$H_{51}: v_5v_1$ & $-4235x-23100y+106722000\ge 0$\\
\bottomrule
\end{tabular}
\end{center}

Intersecting the indicated boundary lines gives the following candidate vertices.  In the table, an active constraint is one attained with equality at the indicated vertex.  Direct substitution verifies that each candidate also satisfies all the remaining half-plane constraints, so these are precisely the vertices of the guard-point set, in cyclic order.
\begin{center}
\begin{tabular}{@{}ll@{}}
\toprule
vertex of $K(F)$ & active constraints\\
\midrule
$(0,-4620)$ & $H_{12},\ H_{23}$\\
$(19800,-990)$ & $H_{23},\ H_{45}$\\
$(22176,0)$ & $H_{34},\ H_{45}$\\
$(19800,990)$ & $H_{34},\ H_{51}$\\
$(0,4620)$ & $H_{51},\ H_{12}$\\
\bottomrule
\end{tabular}
\end{center}
Thus
\[
K(F)=\operatorname{poly}\bigl((0,-4620),(19800,-990),(22176,0),(19800,990),(0,4620)\bigr),
\]
where the displayed order is cyclic.  The shoelace formula, applied in this order, gives
\[
|K|=113430240.
\]
The same formula gives
\[
|F|=115259760,\qquad |C|=115615500.
\]
The side lengths are integral because the relevant difference vectors are scaled Pythagorean triples:
\begin{center}
\begin{tabular}{@{}lclc@{}}
\toprule
segment & component pattern & primitive relation & length\\
\midrule
$v_2v_3,\ v_5v_1$ & $385(\pm60,11)$ & $60^2+11^2=61^2$ & $23485$\\
$v_4v_3,\ v_4v_5$ & $77(12,\pm5)$ & $12^2+5^2=13^2$ & $1001$\\
$v_3v_5$ & $77(0,10)$ & vertical & $770$\\
\bottomrule
\end{tabular}
\end{center}
The perimeter values are
\[
\Per(F)=9240+23485+1001+1001+23485=58212,
\]
\[
\Per(C)=9240+23485+770+23485=56980.
\]
Therefore
\[
G(F)=\frac{113430240}{115259760}=\frac{62}{63},\qquad
E(F)=\frac{115259760}{115615500}=\frac{324}{325},
\]
\[
P(F)=\frac{56980}{58212}=\frac{185}{189}.
\]
Hence
\[
P(F)<G(F)<E(F).
\]

\section{Euclidean estimates for the ordinary perimeter comparison}\label{app:euclidean}
We spell out the compactness and disk-existence points used in \cref{sec:ordinary}.  In this paper a polygon is always the closed polygonal region bounded by a simple polygonal curve; in particular it is compact.  The guard-point set $K(F)$ is the intersection of finitely many closed inner half-planes and is contained in $F$, hence it is a compact convex polygon.  If $|K(F)|=0$, then $G(F)=0$ and the perimeter comparison is immediate.  Assume therefore that $|K(F)|>0$.  The function
\[
x\longmapsto \operatorname{dist}(x,\partial K(F))
\]
is continuous on the compact set $K(F)$, so it attains a maximum $\rho>0$ at some point $o\in K(F)$.  Thus $B(o,\rho)\subset K(F)$.

\par\medskip
\begin{proof}[Proof of \Cref{lem:fan}]
Apply \cref{lem:perAF} with $A=B(o,\rho)\subseteq K(F)$.  Since
\[
h_{B(o,\rho)}(n_i)=o\cdot n_i+\rho\|n_i\|
\qquad\text{and}\qquad
\sum_i n_i=0,
\]
we have
\[
\Per_{B(o,\rho)}(F)
=\sum_i h_{B(o,\rho)}(n_i)
=\rho\sum_i\|n_i\|
=\rho\Per(F).
\]
The generalized anisotropic estimate therefore gives $\rho\Per(F)\le2|F|$, which is the desired inequality.
\end{proof}

\begin{proof}[Proof of \Cref{lem:permonotone}]
First, cutting a compact convex polygonal region by a closed half-plane cannot increase its perimeter, as long as the resulting set has nonempty interior.  Let $Q$ be such a region and let $H$ be a closed half-plane.  If $Q\cap H=Q$, there is nothing to prove.  Otherwise, under the standing assumption that $Q\cap H$ has nonempty interior, the boundary line of $H$ must cross the interior of $Q$ and hence meets $\partial Q$ in two distinct points $X$ and $Y$.  The part of $\partial Q$ removed by the cut is a polygonal chain from $X$ to $Y$; let its side vectors be $w_1,\dots,w_s$.  Then
\[
w_1+\cdots+w_s=Y-X.
\]
By the triangle inequality,
\[
|Y-X|\le |w_1|+\cdots+|w_s|.
\]
The region $Q\cap H$ is obtained from $Q$ by replacing this boundary chain with the segment $XY$, so
\[
\Per(Q\cap H)\le \Per(Q).
\]
Now write $A$ as the intersection of its supporting half-planes,
\[
A=H_1\cap\cdots\cap H_r.
\]
Since $A\subset B$, we also have
\[
A=B\cap H_1\cap\cdots\cap H_r.
\]
Every intermediate intersection $B\cap H_1\cap\cdots\cap H_j$ contains $A$ and therefore has nonempty interior.  Thus the preceding half-plane-cut argument applies at every step.
Applying the preceding half-plane-cut argument successively gives
\[
\Per(A)\le \Per(B).
\]
\end{proof}

\begin{proof}[Proof of \Cref{lem:inradius}]
This estimate is intentionally one-sided in the less familiar direction.  Write the irredundant half-plane description of $Q$ as
\[
Q=\bigcap_j\{x:x\cdot u_j\le b_j\}
\]
with each $u_j$ an outward unit normal.  The set of centers of disks of radius $\rho$ contained in $Q$ is
\[
Q_\rho=\bigcap_j\{x:x\cdot u_j\le b_j-\rho\}.
\]
It is nonempty by the definition of the inradius and has empty interior by its maximality; hence $|Q_\rho|=0$.

Let $x\in Q\setminus Q_\rho$, set
\[
\delta=\min_j\bigl(b_j-x\cdot u_j\bigr)<\rho,
\]
and choose an index $j$ attaining the minimum.  Put $p=x+\delta u_j$.  For every $i$,
\[
p\cdot u_i
=x\cdot u_i+\delta\,u_j\cdot u_i
\le x\cdot u_i+\delta
\le b_i.
\]
Thus $p\in Q$ and $p\cdot u_j=b_j$, so $p$ lies on the side with outward normal $u_j$.  Consequently $x=p-\delta u_j$ belongs to the closed inward rectangle of width $\rho$ based on that side.  The set $Q\setminus Q_\rho$ is therefore covered by the inward rectangles based on all sides of $Q$.  Their total area is $\rho\Per(Q)$, and overlaps can only decrease the area of their union.  Since $|Q_\rho|=0$, we conclude that
\[
|Q|=|Q\setminus Q_\rho|\le \rho\Per(Q).
\]
\end{proof}

\section{A short derivation of the planar Wulff inequality used here}\label{app:wulff}
This appendix records the precise convex-geometric inequality used in \cref{sec:exterior}.  It is included to make the dependence explicit.  The argument uses the planar Brunn--Minkowski theorem and the mixed-area formula for convex polygons to derive exactly the polygonal Wulff inequality needed in the main text.

Let $K,Q\subset\mathbb{R}^2$ be convex polygons, and assume that $K$ has positive area.  For $t\ge 0$, consider the Minkowski sum
\[
Q+tK=\{q+tk:q\in Q,\ k\in K\}.
\]
By the planar mixed-area formula for convex polygons,
\begin{equation}\label{eq:mixed-area-expansion}
|Q+tK|=|Q|+2tV(Q,K)+t^2|K|.
\end{equation}
Here $V(Q,K)$ denotes the mixed area.  With the normalization used in this paper,
\[
2V(Q,K)=\Per_K(Q).
\]
Indeed, after refining the normal fans of $Q$ and $K$ if necessary, the first-order mixed-area term is obtained by moving each supporting side of $Q$ outward in its normal direction by the amount prescribed by the support function of $K$.  If $\ell_i$ and $u_i$ are respectively the side lengths and outward unit normals of $Q$, then
\[
2V(Q,K)=\sum_i \ell_i h_K(u_i)=\Per_K(Q).
\]
Thus \eqref{eq:mixed-area-expansion} becomes
\begin{equation}\label{eq:anisotropic-steiner}
|Q+tK|=|Q|+t\Per_K(Q)+t^2|K|.
\end{equation}
The coefficient $|K|$ in the quadratic term is the mixed-area coefficient of $K$ with itself, equivalently the area term obtained from the homothetic summand $tK$.

The planar Brunn--Minkowski theorem gives
\[
|Q+tK|^{1/2}\ge |Q|^{1/2}+t|K|^{1/2}\qquad(t\ge 0).
\]
Squaring and comparing with \eqref{eq:anisotropic-steiner}, we obtain
\[
|Q|+t\Per_K(Q)+t^2|K|\ge |Q|+2t\sqrt{|Q|\,|K|}+t^2|K|.
\]
For $t>0$, cancellation and division by $t$ yield
\[
\Per_K(Q)\ge 2\sqrt{|Q|\,|K|}.
\]
Therefore
\[
\Per_K(Q)^2\ge 4|K|\,|Q|,
\]
which is the planar Wulff inequality in the form used in the main text.  Equality occurs in the Brunn--Minkowski step when $Q$ and $K$ are homothetic, which gives the usual equality case.

\section*{Declaration of AI-assisted tools}
The author used AI-assisted tools for exploratory discussion, locating related terminology, algebraic checking, and editorial refinement.  The author independently verified the mathematical arguments and computations and takes full responsibility for the content.

\section*{Acknowledgements}
The author is deeply grateful to Thomas Q. Sibley for kindly providing a copy of his paper, and to Joseph O'Rourke for generous and encouraging correspondence.  Their willingness to respond to an independent inquiry was especially helpful in making this note possible.  The author also thanks the organizers of the CCCG 2025 open-problems session and the authors of the open-problems summary through which the question first came to his attention.

\end{document}